\documentclass[journal]{IEEEtran}
\usepackage{amsmath}
\usepackage{amsthm}
\usepackage{graphicx}
\usepackage{hyperref}
\usepackage{cleveref}

\theoremstyle{plain}
\newtheorem{theorem}{Theorem}

\theoremstyle{definition}

\newcommand{\newsection}[1]{\section{#1}} 

\long\def\symbolfootnote[#1]#2{\begingroup%
\def\thefootnote{\fnsymbol{footnote}}\footnotetext[#1]{#2}\endgroup}

\begin{document}

\title{Tighter Worst-Case Bounds on Algebraic Gossip}

\author{Bernhard Haeupler, MIT, haeupler@mit.edu}

\maketitle

\begin{abstract}
Gossip and in particular network coded algebraic gossip have recently attracted attention as a fast, bandwidth-efficient, reliable and distributed way to broadcast or multicast multiple messages. While the algorithms are simple, involved queuing approaches are used to study their performance. The most recent result in this direction shows that uniform algebraic gossip disseminates $k$ messages in $O(\Delta(D + k + \log n))$ rounds where $D$ is the diameter, $n$ the size of the network and $\Delta$ the maximum degree. 

In this paper we give a simpler, short and self-contained proof for this worst-case guarantee. Our approach also allows to reduce the quadratic $\Delta D$ term to $\min\{3n, \Delta D\}$. We furthermore show that a simple round robin routing scheme also achieves $\min\{3n, \Delta D\} + \Delta k$ rounds, eliminating both randomization and coding. Lastly, we combine a recent non-uniform gossip algorithm with a simple routing scheme to get a $O(D + k + \log^{O(1)})$ gossip information dissemination algorithm. This is order optimal as long as $D$ and $k$ are not both polylogarithmically small. 
\end{abstract}

\IEEEpeerreviewmaketitle


\newsection{Introduction}

Broadcast and multicast are fundamental communication primitives with many practical applications like the maintenance of distributed databases or content distribution networks. With modern networks increasing in size they often become less reliable and behave more like distributed systems. This increases the need for reliable ways to disseminate information without knowing the network topology. 

Gossip protocols in which nodes forward information to randomly chosen neighbors have been developed as a powerful new paradigm in this direction. Their randomized approach has been shown to disperse information quickly in many networks while keeping the total number of messages and the congestion on any particular link small. 

When more information than can be fit in a packet is to be distributed selecting the right piece of information to forward can be hard. Algebraic gossip, an adaptation of random linear network coding, has been proposed as a solution and it was shown that in a complete network it outperforms routing \cite{DebM}. Subsequent works like \cite{moskA,borokhovich2010tight,orderopt} extend this idea to general networks. These works also follow the approach of \cite{DebM} to understand the performance of algebraic gossip by analyzing the queuing behavior of innovative packets. The most recent work in this line \cite{borokhovich2010tight,orderopt} use Jackson's theorem to study the worst-case performance of algebraic gossip in terms of network diameter and maximum degree. This will also be the main focus of this paper.

\medskip 
We show that the projection analysis technique of \cite{haeupler2011analyzing} provides easier means to understand algebraic gossip. 
With \Cref{thm:main} we give a simple, short and self-contained proof for the main result in \cite{orderopt}, which itself generalizes the main result in \cite{borokhovich2010tight}. 
In \Cref{thm:maintight} we then further tighten \Cref{thm:main}. Next, we observe that the worst-case performance of algebraic gossip can also be matched by uniform gossip without coding (\Cref{thm:routing2}) or even a simple deterministic round robin routing protocol (\Cref{thm:routing}). Appealing to the optimality of algebraic gossip \cite{ITW_NCoptimality} leads then to yet another alternative proof for Theorems \ref{thm:main} and \ref{thm:maintight}. Lastly, we follow the idea of \cite{orderopt} to bridge the gap to more structured broadcast protocols by giving an efficient non-uniform gossip protocol. Our \Cref{thm:nonuniform} shows that this protocol achieves an optimal performance in any network up to an additive polylogarithmic term. All previous algorithms had gaps of up to $\Theta(n/\log n)$.

\newsection{Gossip Algorithms and the Multicast Problem}

As in \cite{orderopt} we assume a network to be modeled by an undirected graph $G=(V,E)$ with $n=|V|$ nodes, node degree $\Delta_u$ for $u \in V$, maximum degree $\Delta = \max_u \Delta_u$ and diameter $D$. During a {\em gossip protocol} nodes communicate in synchronous rounds in which each node can initiate one bidirectional exchange of a packet with one neighbor. The goal of the {\em information dissemination}, {\em rumor spreading} or $k${\em-message multicast} task is to spread $k$ initially distributed messages to all nodes as fast as possible and with high probability (whp), that is, with probability at least $1-1/n$. We assume that the messages have equal size and that a packet can fit one message plus some header information (which becomes negligibly small for large message/packet sizes). A gossip solution to the multicast problem specifies both, which node contacts which in each round and what is sent in each transmitted packet. The most studied instantiations for these choices are {\em uniform gossip}, that is, each node contacts a uniformly random neighbor in each round, and {\em algebraic gossip}, that is, nodes perform {\em random linear network coding} on the messages to create new packets. In random linear network coding nodes send out random linear combinations of messages (over some finite field $GF(q)$) together with a vector containing all coefficients used. Each new packet is then simply created as a uniformly random linear combination of already received packets or initially known messages; to decode, nodes perform Gaussian elimination.

\newsection{A Simpler and Stronger Proof for Uniform Algebraic Gossip}

The main result in \cite{orderopt} is the following worst-case guarantee on the performance of algebraic gossip in any network:

\begin{theorem}\label{thm:main}
Uniform algebraic gossip over $GF(2)$ disseminates $k$ messages whp in no more than $O(\Delta (D + k + \log n))$ time.
\end{theorem}

The proof is somewhat involved and uses Jackson's queuing theorem. Here, we show that one can alternatively obtain the same result using the projection analysis of \cite{haeupler2011analyzing}. In particular, one can essentially get it by setting the min-cut probability $\gamma$ in \cite[Lemma 6]{haeupler2011analyzing} to be $1/\Delta$. Next, we give a simple, short and completely self-contained proof based on \cite[Theorem 3]{haeupler2011analyzing}:

\begin{proof}
We say a node knows a coefficient vector $\mu \in \{0,1\}^k$ if it has received a packet with a coefficient vector that is non-perpendicular to $\mu$ (over $GF(2)$). We claim that for any non-zero vector $\mu$ the probability that any fixed node $v$ does not learn $\mu$ within $O(\Delta(D + k + \log n))$ rounds is at most $2^{-(k + 2 \log n)}$. Then, a union bound over all nodes and all $2^k$ coefficient vectors shows that whp all nodes know all vectors. From this it is easy to conclude that all nodes can decode. 

To prove this claim we look at a shortest path $P$ from $v$ to a node that initially knows $\mu$ (that is, that starts with a message with a non-zero coefficient in $\mu$). At any round $t$ let node $u$ be the closest node to $v$ on $P$ that knows $\mu$. There is a $1/\Delta_u$ chance that $u$ contacts the next node on the path and independently a chance of $1/2$ that $u$ sends out a packet with a coefficient vector that is non-perpendicular to $\mu$. Thus, in any round independently with probability at least $\frac{1}{2\Delta}$ knowledge of $\mu$ makes progress on $P$. A Chernoff bound shows that the probability that less than $D$ progress is made in $16\Delta(D + k + \log n)$ rounds when $8(D + k + \log n)$ successes are expected is at most $2^{-(k + 2 \log n)}$ as claimed.
\end{proof}

\smallskip

Note, that even for constant $k$ \Cref{thm:main} results in a quadratic bound of $\Delta D = \Theta(n^2)$ for many networks. Next, we tightening the bound of \Cref{thm:main} and show that small number of messages (e.g., $k = O(\log n)$) never take more than $O(n \log n)$ time:

\begin{theorem}\label{thm:maintight}
Uniform algebraic gossip disseminates $k$ messages whp in at most $O(\min\{n, \Delta D\} + \Delta (k + \log n))$ time.
\end{theorem}
\begin{proof}
In the proof of \Cref{thm:main} let $X_1$, $X_2, \ldots$ be the number of rounds spent for knowledge about $\mu$ to successfully make the $i$th step along $P$ towards $v$ and let $X = \sum_i X_i$ be the total number of rounds needed. Note that each $X_i$ is an independent geometric random variable with success probability $\frac{1}{2\Delta_u}$ and thus $E[X] = \sum_i E[X_i]
 = \sum_{u \in P} 2 \Delta_u \leq 2 D \Delta$. Since $P$ is a shortest path, the sum of node degrees along $P$ is at most $3n$ because every node can be adjacent to at most three consecutive nodes on $P$ without creating a shortcut. We thus also have $E[X] \leq 6n$. By applying Markov's inequality to $(1 - 1/\Delta)^{-2X}$ one can show that $P[X \geq 2(E[X] + t)]< (1 - 1/\Delta)^{-t}$. Using this with $t = \Delta (k + \log n)$ replaces the Chernoff bound argument in the second paragraph of the proof of \Cref{thm:main}.
\end{proof}

\newsection{Gossip without Coding or Randomization}

Among the main motivations for using randomized gossip is that its randomization adds robustness and often avoids bottlenecks in the network. In regular, well-connected networks this allows for much faster dissemination times than the bound given in \Cref{thm:maintight}. Adding coding on top of this (in particular, network coded algebraic gossip) allows mixing and efficient diffusion of information which often significantly increases throughput~\cite{haeupler2011analyzing}. 

In this section, we show that if one just wants to achieve the worst-case running time of \Cref{thm:maintight} neither randomization nor coding is needed. Instead the following simple deterministic {\em round robin routing} achieves the same performance (with no probability of failure): Each node $u$ repeatedly contacts its neighbors such that every neighbor is contacted at least every (order) $\Delta_u$ rounds. A node furthermore can forward any packet as long as the same packet is not forwarded twice to the same node. One simple protocol that does this efficiently is {\em prioritized round robin routing}. In this protocol each node fixes a cyclic order of its neighbors and simply contacts the next node in its list in every round. Each node $u$ furthermore keeps a counter how often it has sent out each packet and sends out the lexicographically first packet it has not sent out $\Delta_u$ often already (if such a packet exists).

\begin{theorem}\label{thm:routing}
(Prioritized) round robin routing disseminates all $k$ messages in at most $\min\{3n, \Delta D\} + \Delta k$ time. 
\end{theorem}
\begin{proof}
Throughout the proof we use the following notion of distance between any two nodes $u$ and $v$: Let $dist(u,v)$ be the smallest number such that there is a path $P$ between $u$ and $v$ with $\sum_{w \in P-\{v\}} \Delta_w \leq dist(u,v)$. Note that $dist(u,v) \leq \min\{3n, \Delta D\}$ as shown in the proof of \Cref{thm:maintight}. We claim that for any nodes $u$,$v$, any message $m$ and any $i \geq 0$ if node $u$ knows message $m$ initially then at time $t = dist(u,v) + \Delta i$ node $v$ will either know message $m$ or at least $i+1$ messages smaller than $m$. Note that this claim also implies the theorem since at time $\min\{3n, \Delta D\} + \Delta k$ any node has either received every message $m$ or $k$ (smaller) messages.

We prove the claim by induction on $t$. Nothing needs to be shown for $t=0$. For the induction step consider any $u,v,m$ and $t = dist(u,v) + \Delta i$ and let $v'$ be the neighbor of $v$ on a shortest distance path from $u$ to $v$. By induction hypothesis we get that for any $j \leq i$ node $v'$ at time $d(u,v') + \Delta j$ either knows $m$ or $j+1$ messages smaller than $m$. From this it is clear that at time $d(u,v') + \Delta_v' = d(u,v)$ node $v'$ has sent either message $m$ or one message smaller than $m$ to $v$ and similarly for any $j<i$ at time $d(u,v') + \Delta j + \Delta_v = d(u,v) + \Delta j$ node $v$ has sent either $m$ or $j+1$ messages to $v$ that are smaller than $m$. For $j=i$ this is exactly the claim that needed to be proved. 
\end{proof}

\smallskip

We remark that a similar performance guarantee also holds for randomized uniform gossip. It can be obtained along the same line as the proof for \Cref{thm:routing}. We state the result next but omit the proof. 

\begin{theorem}\label{thm:routing2}
Prioritized uniform gossip whp routes all messages to all nodes in $O(\min\{n, \Delta D\} + \Delta (k + \log n))$ time.
\end{theorem}

Given that these two simple routing schemes achieve a good worst-case performance one would expect that adding algebraic gossip on top of it should not harm the performance. Indeed we can make this formal and appeal to the optimality of algebraic gossip~\cite{ITW_NCoptimality} to give yet another simple and alternative proof for \Cref{thm:maintight}:

\smallskip 

\begin{proof}[Proof of \Cref{thm:maintight}]
In contrast to before we require the field size used for coding to be at least $q = n^2$ which gives rise to $2\log n$-size coefficients. In \cite{ITW_NCoptimality} it is shown that with probability $1 - n/q$ algebraic gossip completes in exactly optimal time, that is, at the first time it is possible in hindsight to route all messages to each node individually. Furthermore, \Cref{thm:routing2} shows how to route the messages whp in $O(\min\{3n, \Delta D\} + \Delta (k + \log n))$ rounds via the uniform gossip exchanges. Thus with probability $1 - n/q - 1/n = 1 - 2/n$ uniform algebraic gossip also completes in this time. 
\end{proof}

\newsection{Faster Non-Uniform Gossip}

Lastly, \cite{orderopt} shows that fast non-uniform gossip protocols can be obtained by first running a $1$-message broadcast gossip algorithm and then performing algebraic gossip along an induced spanning-tree. More specifically, for any broadcast algorithm with running time $B$ (and given a leader) this leads to a $O(B + k + \log n)$ solution to the $k$-message multicast problem considered here. Instantiating this with the broadcast protocol from~\cite{weakconductance} it is shown that an $\Theta(D + k + \Phi^{-1}\log n)$ algebraic gossip protocol is possible for graphs with weak conductance $\Phi$. Unfortunately, there are graphs for which $\Phi^{-1}\log n = \Theta(n)$ even so the network diameter is logarithmic (for example, a balanced binary tree). For these graphs a gap of up to $\Theta(n/\log n)$ between the upper bound and the $\Omega(D + k)$ lower bound remains. The following lemma improves upon this using~\cite{gossip}, a recent strengthening of~\cite{weakconductance}:

\begin{theorem}\label{thm:nonuniform}
There is a gossip algorithm that whp routes $k$ messages to all nodes in $O(D + k + \log^{O(1)})$ time in any network. 
\end{theorem}
\begin{proof}
In~\cite{gossip} a broadcast gossip algorithm with time $O(D + \log^{O(1)})$ is given. We first run this algorithm with each node using its ID as a message. After completion every node declares the node it first learned the smallest ID from as its parent (breaking ties arbitrarily). This induces a spanning tree. In all future rounds, each node initiates a bidirectional exchange with its parent forwarding any message not sent by it before (if such a message exists). A standard pipelining proof (similar to \Cref{thm:routing}) shows that after $O(k + D')$ time all nodes know about all messages, where $D'$ is the diameter of the tree. Since a message can travel at most one step per round during the initial broadcast we have $D' = O(D + \log^{O(1)})$ which completes the proof.
\end{proof}

\newsection{Conclusion}

We have given simple, short and tighter proofs for the worst-case performance of algebraic gossip using the projection analysis technique of \cite{haeupler2011analyzing}. This improves over the recent results in \cite{orderopt}. We could furthermore show that this worst-case performance is also met by a simple round robin routing scheme that eliminates both randomness and coding. A similar routing scheme has been studied intensively under the term quasirandom rumor spreading \cite{doerr2008quasirandom} and has been shown to perform surprisingly well in many topologies. Lastly, bridging the gap to more structured broadcast protocols we have shown a non-uniform gossip protocol that achieves order optimal (up to an additive polylogarithmic term) distributed information spreading. All these results do not take any unreliability or changes in the network topology into account. This is unfortunate since fault-tolerance with respect to dynamic or unreliable topologies is one of the main motivations behind pursuing randomized gossip approaches. Some work in this direction can be found in \cite{haeupler2011analyzing} but any further approaches to formalize and study dynamic topologies and the reliability of the algorithms presented here is an important open problem of interest.


\begin{thebibliography}{1}

\bibitem{DebM}
S.~Deb, M.~M{\'e}dard, and C.~Choute,
\newblock ``Algebraic gossip: a network coding approach to optimal multiple
  rumor mongering,''
\newblock {\em IEEE/ACM Trans. Netw.}, vol. 14, no. SI, pp. 2486--2507, June
  2006.

\bibitem{moskA}
D.~Mosk-Aoyama and D.~Shah,
\newblock ``Information dissemination via network coding,''
\newblock in {\em IEEE International Symposium on Information Theory}. IEEE,
  2006, pp. 1748--1752.

\bibitem{borokhovich2010tight}
M.~Borokhovich, C.~Avin, and Z.~Lotker,
\newblock ``Tight bounds for algebraic gossip on graphs,''
\newblock in {\em IEEE International Symposium on Information Theory}, 2010,
  pp. 1758--1762.

\bibitem{orderopt}
C.~Avin, M.~Borokhovich, K.~Censor-Hillel, and Z.~Lotker,
\newblock ``Order optimal information spreading using algebraic gossip,''
\newblock in {\em PODC}, 2011, pp. 363--372.

\bibitem{haeupler2011analyzing}
B.~Haeupler,
\newblock ``Analyzing network coding gossip made easy,''
\newblock in {\em STOC}, 2011, pp. 293--302.

\bibitem{ITW_NCoptimality}
B.~Haeupler, M.~Kim, and M.~Medard,
\newblock ``{Optimality of Network Coding with Buffers},''
\newblock in {\em ITW}, 2011, pp. 533--537.

\bibitem{weakconductance}
K.~Censor-Hillel and H.~Shachnai,
\newblock ``Fast information spreading in graphs with large weak conductance,''
\newblock in {\em SODA}, 2011, pp. 440--448.

\bibitem{gossip}
K.~Censor-Hillel, B.~Haeupler, J.~Kelner, and P.~Maymounkov,
\newblock ``{Global Computation in a Poorly Connected World: Fast Rumor
  Spreading with No Dependence on Conductance},''
\newblock in {\em STOC}, 2012.

\bibitem{doerr2008quasirandom}
B.~Doerr, T.~Friedrich, and T.~Sauerwald,
\newblock ``Quasirandom rumor spreading,''
\newblock in {\em ACM-SIAM Symposium on Discrete Algorithms}, 2008, pp.
  773--781.

\end{thebibliography}
\end{document}